\documentclass[a4paper]{article}
\usepackage{graphicx}
\usepackage{url} 
\usepackage{hyperref}
\usepackage{a4wide}
\usepackage{amsmath,amsfonts,amssymb,amsthm}
\newcommand{\R}{{\mathbb R}}

\newcommand{\y}{{\bf y}}

\newcommand{\x}{{\bf x}}

\newcommand{\A}{{\bf A}}

\newcommand{\J}{{\bf J}}

\newtheorem{theorem}{Theorem}[section]

\newtheorem{corollary}[theorem]{Corollary}
\newtheorem{lemma}[theorem]{Lemma}
\newtheorem{assumption}[theorem]{Assumption}
\theoremstyle{definition}

\numberwithin{equation}{section}

\begin{document}

\noindent 
\begin{center}
\textbf{\large On the transfer matrix of a MIMO system}
\end{center}

\begin{center}
December 15, 2010
\end{center}

\vspace{0.5cm}

\noindent 

\begin{center}
\textbf{ 
Fran{\c c}ois Bentosela\footnote{Centre de Physique Th{\' e}orique,  
    Campus de Luminy, Case 907 - 13288 Marseille cedex 9, France}$^{,2}$, 
Horia D. Cornean\footnote{Department of Mathematical Sciences, 
    Aalborg
    University, Fredrik Bajers Vej 7G, 9220 Aalborg, Denmark; email: cornean@math.aau.dk},
Bernard H. Fleury\footnote{Department of Electronic Systems, 
    Aalborg
    University, Fredrik Bajers Vej 7A, 9220 Aalborg, Denmark},
Nicola Marchetti\footnote{Department of Electronic Systems, 
    Aalborg
    University, Fredrik Bajers Vej 7A, 9220 Aalborg, Denmark}
}
     
\end{center}

\noindent

\begin{abstract}
We develop a deterministic ab-initio model for the
input-output relationship of
a multiple-input multiple-output (MIMO) wireless channel, starting
from the Maxwell equations combined with Ohm's Law. 
The main technical tools are  scattering and geometric 
perturbation theories. 
The derived relationship can lead us to a deep understanding of how 
the propagation conditions and the coupling effects between 
the elements of multiple-element arrays affect the properties 
of a MIMO channel, e.g. its capacity and its number of degrees of freedom. 
\end{abstract}

\tableofcontents

\vspace{0.5cm}

\section{Introduction}

In wireless communication, systems with multiple antenna arrays at both ends of the link are called MIMO (multiple input multiple output). The main interest is to compute the number of bits/s which can be transmitted  by these systems. It has been shown that this quantity, called  channel capacity \cite{Fo,Te,P,M,Al}, is linked with the transfer matrix $\mathcal{H}$ whose element $\mathcal{H}_{mn}$ is the ratio between the complex current intensity $I_m$ at the load of the $m^{th}$ receiving (RX) antenna over the current $I_n$ which feeds the $n^{th}$ transmitting (TX) antenna.
If $M$ and $N$ are respectively the number of   $RX$ antennas and  $TX$ antennas,  then $\mathcal{H}$ is a $M\times N$ matrix. 

\noindent Even if one can experimentally determine $\mathcal{H}$  in a given environment when the antennas are fixed, its entries show such large variations when moving the antennas arrays from place to place
that the knowledge of $\mathcal{H}$ for some number of locations seems to be insufficient for having an idea of the capacity everywhere. Fortunately, it seems that the eigenvalue distribution of the matrix $\mathcal{H}^*\mathcal{H}$ (here ${}^*$ means adjoint) is
more robust with respect to  the displacements, but this has to be better understood.
In order to mimic the variations of the entries of $\mathcal{H}$, one possibility is to replace $\mathcal{H}$ by random matrices introducing several distributions for their  elements (see \cite{H,Z}).
In sharp contrast with this method, we try to understand the variations through the physical study of the wave propagation in a given environment. Thus we employ a deterministic modeling method and our aim is to describe $\mathcal{H} $ by quantities linked with the physical characteristics of the antennas and the environment. But certain stochastic aspects can naturally appear if some parameters of the scatterers present in the medium are described by stochastic variables (like position, density etc).

\noindent This work was motivated by the existence of an heuristic formula intensively used by the MIMO signal experts. It appeared in the works of B.Fleury et al. see e.g. \cite{F1,F2,F3,F4}; see also \cite{S}. This formula describes the transfer matrix in the following way:
\begin{align}\label{februar13}
\mathcal{H}= \int_{S^2\times S^2}
 c(\Omega_R)^t 
 \mathcal{A}(\Omega_R,\Omega_T)
c(\Omega_T)
 d\Omega_R d\Omega_T .
\end{align}
where 
$\Omega_T=(\phi_T,\theta_T )$ is a direction of departure from  the $TX$ side where $\phi_T,\theta_T$ are respectively its azimuth and elevation 
angles and $\Omega_R=(\phi_R,\theta_R)$ is a direction of arrival at  the $RX$ side,
$c(\Omega_T)$ and $c(\Omega_R)$ are  respectively 
$2\times N$ and $2\times M$ matrices whose elements depend only on the corresponding antennas diagrams and the geometry of the arrays, while $\mathcal{A}(\Omega_R, \Omega_T)$ is a $2\times 2$
matrix which depends only on the environment.

\noindent Our aim is to understand this formula from first principles, that is  solving  the Maxwell equations in a typical  complex environment formed by the  $TX$  and $RX$ antenna arrays surrounded by buildings , persons, trees ...
 The antennas are pieces of metal described 
by their geometry and by their conductivity $\sigma$, while the scatterers which are made of dielectric and magnetic materials can also be described by their geometry, permittivity,
permeability and conductivity.

\noindent To determine $\mathcal{H}_{mn}$ one introduces  
in the second member of the Maxwell equations a harmonic current density $\tilde{\J}_{in}$, in a small  volume constituted  
by a part of the wire which connects the $n^{th}$ signal source to the  antenna. The current intensity crossing a section $S_T$  of the wire  is $I_n=\int_{S_T}\tilde{\J}_{in}(x)\cdot d{\bf S}_T$.
If one knows the electric fields ${\bf E}$  induced by this current density, in particular at a section $S_R$ 
of the wire connecting the $m^{th}$ RX antenna to its load , there 
one can calculate the current density using the Ohm law  
$\tilde{\J}_{out}= \sigma {\bf E}$ and deduce the
current crossing this section,   
$I_m=\sigma\int_{S_R}{\bf E}(x). d{\bf S}_R$. Then $\mathcal{H}_{mn}=\frac{I_m}{I_n}$; note that this formula is 
true when only the $n$'th emitting antenna is fed with current. For the general case see the discussion around 
\eqref{februar6}. 

\noindent To tackle these calculations, wave propagation experts 
try to address the difficulties in the following manner. They suppose that they know the electromagnetic fields (EM) fields radiated by  the $n^{th}$ $TX$ antenna in the free space. These fields are considered as incoming fields in a scattering problem where the scatterers are constituted by buildings, persons, trees... and they try to calculate the total fields produced adding the incoming fields and the scattered fields. This is in fact a difficult problem which is only solved analytically in some very simple and unrealistic situations. Practically it is solved using the geometrical optics approximation 
(ray tracing). Finally, the total $EM$ fields calculated in this way are considered
as incoming fields for the $m^{th}$ $RX$ antenna and, once more, one is faced with a scattering problem where now the $RX$ antennas play the role of scatterers and again one has to determine the new total fields, from which one obtains $I_m$ .

\noindent One of the purposes of the paper is to justify this procedure.
We succeed in proving that when the antennas 
are sufficiently small and spatially separated from the surrounding scatterers, the $EM$ fields calculated with the procedure explained above appear as the 
first term in a series giving the exact fields. In Lemmas \ref{lemma1} and \ref{lemma2}  we give an approximation for 
the total Green function which puts into evidence the decoupling between the $TX$ antennas, the $RX$ antennas and the surrounding scatterers.

\noindent Our proposed methodology is to use the vector 
potential formalism to
solve the Maxwell equations.
Our main mathematical tool is geometric 
perturbation theory for closed but non-selfadjoint
operators. These methods have been developed for apparently unrelated 
quantum scattering problems in \cite{He, B, C1, C2} and they 
can be used for describing the propagation of electromagnetic waves in frequency domain.

\noindent In Corollary \ref{corollary1} we derive a stronger version of the 
 Fleury heuristic formula as a corollary of Theorem \ref{theorem1}, in the case when the antenna arrays are far away from the surrounding scatterers. In particular, we link the heuristic "spread matrix function" 
$\mathcal{A}(\Omega_R,\Omega_T)$ to the Green function of the environment alone.

\noindent In this  paper  the language is deliberately made free of too technical notions of functional analysis and operator theory in order to facilitate the understanding for non-mathematicians. Some notions concerning the Limiting Absorbtion Principle for  the non self adjoint operators coming from the vector potential description of the $EM$ fields are not yet proved, so this paper  has to be completed in this respect in order to get rigorously the set of frequencies for which the Green operators exist when 
the "energy" tends to a real number, as operators in certain weighted $L^{2}$ spaces.

\noindent On the other hand, this paper is a starting point to more applied engineering works  since now we understand 
the link between the heuristic "spread function" in formula (\ref{februar13}) and the
scattering transfer operator of the environment.

\subsection{Main technical assumptions}

Throughout this paper the configuration space will be
$\mathbb{R}^3$. The environment is known, which means that we assume
the knowledge of the dielectric constant $\epsilon$ and the
conductivity $\sigma$.  They are smooth functions of the position and
frequency. There are no magnetic effects, that is $\mu=\mu_0$ is constant. 

\noindent Let us assume that we know the current density 
$\tilde{\J}_{in}(\x,t)$ (i.e. in time domain) 
at the entrance of the transmitting antenna. 
It is completely determined by the transmission device. 
\begin{assumption}\label{Asump1}
The current density $\tilde{\J}_{in}$ has the following generic properties:
\begin{itemize}
 \item it is smooth in the time variable; 
\item it is compactly supported in $\x$;
\item $\frac{1}{\sqrt{2\pi}}\int_{\R}\tilde{\J}_{in}(\x,t)e^{-j\omega
    t}dt=:\J(\x,\omega)=
\overline{\J(\x,-\omega)}$ is a smooth function, compactly supported
in $\omega\in [-\omega_2,-\omega_1]\cup [\omega_1,\omega_2]$. This
equality is a consequence of the reality of $\tilde{\J}_{in}(\x,t)$. 
\end{itemize}
\end{assumption}

\noindent There is an inherent physical imprecision on the very notion of "external"
source current, since this quantity is supposed to be completely
determined by the transmitter.  It is common to assume that the current densities on
the wires feeding the transmitting antennas  are constant on the disc of the corresponding 
cross sections. But this fact does not affect the conclusions of our
paper.

\begin{assumption}\label{Asump2}
Further technical assumptions on the environment:
\begin{itemize}
 \item $\epsilon(\x,\omega)=\overline{\epsilon(\x,-\omega)}$ and 
$\sigma(\x,\omega)=\overline{\sigma(\x,-\omega)}$ are smooth in $\x$. 
The case with piecewise constant $\epsilon$ requires a rather
complicated regularization procedure which will be considered elsewhere;

\item the dielectric constant of the air is constant and
  equals $\epsilon_0$;
\item 
\begin{align}\label{prima1}
\epsilon^{(r)}(\x):= \frac{\epsilon(\x)}{\epsilon_0}=1+\delta\epsilon_T(\x)+
\delta\epsilon_M(\x) +\delta\epsilon_R(\x),
\end{align}
where all the relative $\delta\epsilon$'s are compactly supported
perturbations, with disjoint supports. Here $T$ means the 
transmitting antenna(s), 
$R$ the receiver(s), and $M$ the scatterer(s).
\item the conductivity can be written
\begin{align}\label{adoua1}
 \sigma(\x)=\sigma_T(\x)+\sigma_R(\x),
\end{align}
where the $\sigma$'s are compactly supported on the regions containing
the antennas. 
\end{itemize}
\end{assumption}

\section{The Maxwell equations in the frequency domain}

In the frequency domain, 
the electric and magnetic fields ${\bf E}$ and 
${\bf H}$ satisfy the
following system of equations:
\begin{align}\label{aunspea2}
 \nabla\cdot {\bf H}&=0;\\
\nabla\times {\bf H}&={\bf J}+(j\epsilon\omega+
\sigma){\bf E}\label{aunspea3};\\
\nabla\times {\bf E}&=-j\mu_0\omega {\bf H}\label{aunspea4};\\
\rho&=\nabla\cdot(\epsilon {\bf E}); \label{februar1} \\
-j\omega \rho&=\nabla\cdot ({\bf J}+\sigma {\bf E}).\label{februar2}
\end{align}
Here we incorporate Ohm's Law in \eqref{aunspea3}, 
in the sense that an electric field
${\bf E}$ generates a charge-current density $\sigma{\bf E}$ in 
metals. We stress that we do not work with perfect conducting antennas;
taking $\sigma$ to infinity is a singular operation within our
formalism, and the study of this limit is a very interesting
mathematical problem which is left as an open problem. We note that
one can start with infinite conductivities right from the beginning,
but the price is that one needs to deal with severely ill-posed
inverse problems (see e.g. \cite{O, AKL}). 
 
\subsection{The magnetic vector potential method}

Equations \eqref{februar1} and \eqref{februar2} allow us to eliminate 
the unknown charge density $\rho$: 
\begin{equation}\label{aunspea1}
 \nabla\cdot\{{\bf J}+(j\epsilon\omega+\sigma){\bf E}\}=0.
\end{equation}
Equation \eqref{aunspea2} allows us to represent 
${\bf H}=\nabla \times {\bf A}$, where ${\bf A}$ is a 
magnetic vector potential. Using this in \eqref{aunspea4} gives 
$\nabla\times ({\bf E}+j\mu_0\omega {\bf A})=0$, thus there 
exists a scalar potential $\phi$ such that:
\begin{equation}\label{aunspea5}
{\bf E}+j\mu_0\omega {\bf A}=-\nabla \phi.
\end{equation}
Introducing the magnetic vector potential in \eqref{aunspea3} and
using \eqref{aunspea5} we obtain: 
\begin{equation}\label{aunspea6}
\nabla (\nabla \cdot {\bf A})-\Delta {\bf A}={\bf J}+
(j\epsilon\omega+\sigma)(-j\mu_0\omega {\bf A}-\nabla \phi).
\end{equation}
Note the important thing that if we have a pair $({\bf A},\phi)$ which
solves \eqref{aunspea6}, then \eqref{aunspea1} is automatically
satisfied because the left hand side of \eqref{aunspea6} is divergence
free. Thus we have the freedom of choosing a $\phi$ completely determined by ${\bf
  A}$ through a Lorenz-type gauge condition, and then solve \eqref{aunspea6} for ${\bf A}$. 
Denoting by ${\bf A}_L$ such a special solution, we choose:
\begin{equation}\label{aunspea7}
\phi_{L}:=-\frac{1}{j\epsilon\omega+\sigma}(\nabla \cdot {\bf A}_L).
\end{equation}
With this choice, \eqref{aunspea6} simplifies to:
\begin{equation}\label{aunspea8}
-\Delta {\bf A}_L={\bf J}-j\mu_0\omega
(j\epsilon\omega+\sigma){\bf
  A}_L-(\nabla \cdot {\bf A}_L)\frac{\nabla(j\epsilon\omega+\sigma)}
{(j\epsilon\omega+\sigma)}.
\end{equation}
Introducing the notation
\begin{equation}\label{aunspea9}
k^2(\x):=\mu_0
\epsilon(\x)\omega^2-j\mu_0\omega \sigma(\x),
\end{equation}
we can rewrite \eqref{aunspea8} as:
\begin{equation}\label{aunspea10}
-\Delta {\bf A}_L-k^2{\bf A}_L+(\nabla \cdot {\bf A}_L)\nabla
\ln(k^2)={\bf J}.
\end{equation}
Now if we have a solution to \eqref{aunspea10}, then the fields ${\bf
  E}$ and ${\bf H}$ are given by:
\begin{align}\label{aunspea11}
{\bf E}&=-j\mu_0\omega {\bf A}_L+\nabla \left (\frac{\nabla \cdot {\bf
      A}_L}{j\epsilon\omega+\sigma}\right ),\\
{\bf H}&=\nabla\times {\bf A}_L\nonumber .
\end{align}

\noindent Denoting by $k_0^2:=\mu_0 \epsilon_0\omega^2$ and using 
$\epsilon^{(r)}(\x)=\epsilon(\x)/\epsilon_0$ (see \eqref{prima1}) we have: 
\begin{equation}\label{aunspea12}
k^2(\x)=k_0^2+(\epsilon^{(r)}(\x)-1)k_0^2-j\mu_0\omega
\sigma(\x)=:k_0^2+\delta k^2(\x),
\end{equation}
where $\delta k^2(\x)=(\epsilon^{(r)}(\x)-1)k_0^2-j\mu_0\omega
\sigma(\x)$ is smooth and compactly supported. Note that
$k^2$ cannot be zero if $\epsilon^{(r)}\neq 0$. 

Define the following first order differential operator $W$ in the
following way:
 \begin{equation}\label{aunspea13}
W{\bf A}:=-\delta k^2 {\bf A} +\{\nabla \cdot {\bf A}\} \nabla \ln(k^2) 
\end{equation}
and observe that it is relatively compact with respect to $-\Delta$ on
the Hilbert space $[L^2(\R^3)]^3$. Moreover, the 
operator sum $-\Delta +W$ is closed on the Sobolev space $[H^2(\R^3)]^3$,
and if $z\in \mathbb{C}$ then we can define the inverse $(-\Delta
+W-z)^{-1}$ for at least when $|Im(z)|$ is large enough. 

\subsection{The transfer matrix}
The solution we are looking for in \eqref{aunspea10} is:
 \begin{equation}\label{aunspea14}
{\bf A}_L=(-\Delta
+W-k_0^2-j0_+)^{-1} {\bf J}_{\rm in},
\end{equation}
where we have to give a rigorous sense to the limiting absorption
principle. The problem is far from being trivial because the perturbation $W$ is
not symmetric; technical details and full proofs will be given
elsewhere. We only want to mention that under rather general
conditions, the limit in \eqref{aunspea10} is expected to hold outside
a discrete set of frequencies. In order to simplify the notation, we
stop writing $0_+$.  

\noindent Having determined ${\bf A}_L$, we can derive ${\bf E}$ from
\eqref{aunspea11}: 

\begin{equation}\label{februar5}
{\bf E}=-j\mu_0\omega(-\Delta
+W-k_0^2)^{-1} {\bf J}_{\rm in}+\nabla \left (\frac{1}{j\epsilon\omega+\sigma}\nabla \cdot(-\Delta
+W-k_0^2)^{-1} {\bf J}_{\rm in} \right ).
\end{equation}

\noindent Let us be more precise and assume that we have $N\geq 1$ transmitting 
antennas for which we can write:
\begin{equation}\label{februar3}
{\bf J}_{\rm in}=\sum_{n=1}^N I_{\rm in}^{(n)} {\bf J}_{\rm in}^{(n)},
\end{equation} 
where $I_{\rm in}^{(n)} (\omega)$ is the current running through the tranverse
section of the $n$-th antenna, and ${\bf J}_{\rm in}^{(n)}$ is its 
corresponding normalized current density (i.e. the flux integral of
any ${\bf J}_{\rm in}^{(n)}$ through the transverse section of antenna
$n$ equals one). 

\noindent According to \eqref{adoua1}, the receiving region has a conductivity
$\sigma_R$. Let us assume that there are $M\geq 1$ disjoint receiving antennas
so that we can write:
\begin{equation}\label{februar4}
\sigma_R=\sum_{m=1}^M \sigma_R^{(m)}.
\end{equation} 
Via Ohm's Law ${\bf J}_{\rm out}=\sigma_R
{\bf E} $ we determine the current
density at the $m$-th receiving antenna. Assume that $S_m$ is the
transverse section of the $m$-th receiving antenna. Then the current
induced in it will be:

\begin{align}\label{februar6}
I_{\rm out}^{(m)}&=\int_{S_m} \sigma_R^{(m)}(\x)
{\bf E}(\x)\cdot d{\bf S}=\sum_{n=1}^N\mathcal{H}_{mn}I_{\rm in}^{(n)},\nonumber
\\
\mathcal{H}_{mn}=&\int_{S_m} \sigma_R^{(m)}(\x)\left \{ -j\mu_0\omega(-\Delta
+W-k_0^2)^{-1} {\bf J}_{\rm in}^{(n)}\frac{}{} \right . \nonumber \\ 
&\left . \frac{}{}+\nabla \left
  (\frac{1}{j\epsilon\omega+
\sigma_R^{(m)}}\nabla \cdot(-\Delta
+W-k_0^2)^{-1} {\bf J}_{\rm in}^{(n)} \right )
\right\}\cdot d{\bf S}.
\end{align}

\noindent The transfer matrix elements $\mathcal{H}_{mn}$ are only frequency dependent
and give the current induced in the $m$-th receiving antenna when only
the $n$-th transmitting antenna is fed with current. One must note
that the above formula takes into consideration all possible
couplings. In the rest of the paper we will try to reduce the
complexity of this formula and to arrive to a simpler and practically
more convenient expression.

\section{Decoupling the receivers and transmitters from the environment}

It is important to realize that the resolvent $(-\Delta
+W-k_0^2)^{-1}$
contains the whole information about both the electric and magnetic
fields. 
We will express this resolvent and thus $\mathcal{H}_{mn}$ in a different way, which 
shows a decoupling between transmitters, receivers and scatterers.  

\noindent Let us introduce some notation which would describe a situation 
in which only the transmitter would be present:
\begin{equation}\label{aunspea15}
k_T^2(\x)=k_0^2+\delta\epsilon_T(\x)k_0^2-j\mu_0\omega
\sigma_T(\x)=:k_0^2+\delta k_T^2(\x),
\end{equation}
where $\delta k_T^2$ is again smooth and supported only near the
transmitter(s). In a similar way we introduce the corresponding
quantities for $M$ and $R$. 

\noindent The perturbation corresponding only to the transmitter(s) $W_T$ is:
 \begin{equation}\label{aunspea16}
W_T{\bf A}:=-\delta k_T^2 {\bf A} +\{\nabla \ln(k_T^2)\} \{\nabla \cdot {\bf A}\} 
\end{equation}
and similar objects can be defined for $R$ and $M$. 

\noindent Let us introduce the following operators:
\begin{equation}\label{adoua10}
H=-\Delta +W_T+W_M+W_R,\quad 
H_T:=-\Delta+W_T, \quad H_M :=-\Delta+W_M, \quad H_R:=-\Delta+W_R.
\end{equation}
Here $H_T$ only takes into consideration the perturbation induced by
the transmitter(s), 
$H_M$ does the same thing for the environment, and $H_R$ for 
the receiver(s). Assume that both the transmitter(s) and receiver(s) are
separated from all other scatterers, and from each other (see figure 1
for what follows). 

\begin{center}
{\rm Figure 1.}

            \includegraphics[scale=0.5]{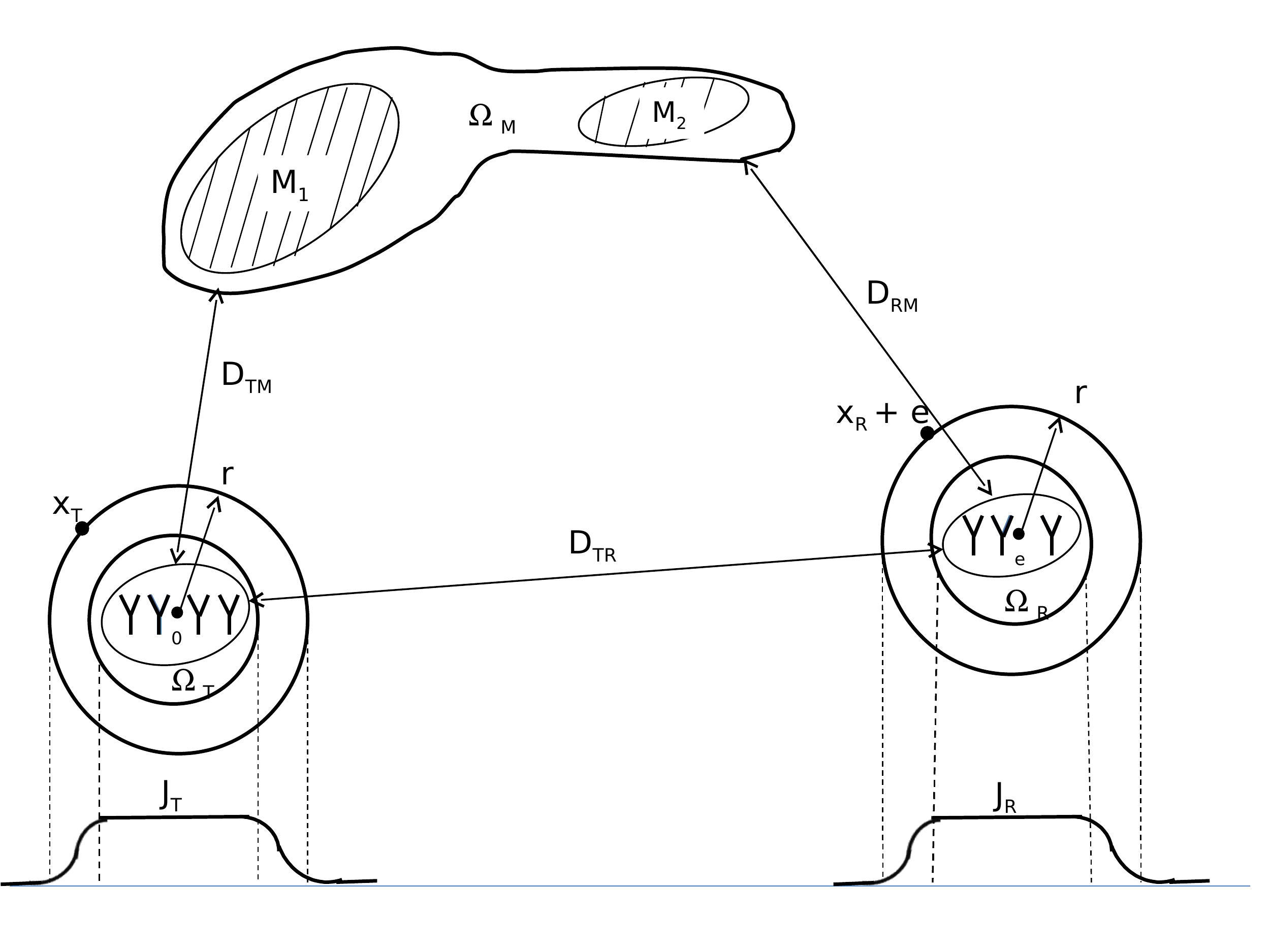} 

\end{center}

\noindent If
$\Omega_{T(RM)}$ are bounded open domains completely containing 
the supports of $W_{T(RM)}$, then we quantify
this separation by:
\begin{align}\label{atreia10}
 &{\rm dist}[\Omega_T,\Omega_R]=D_{TR}>0, \quad 
\min \left\{{\rm dist}[\Omega_R,\Omega_M],
{\rm dist}[\Omega_T,\Omega_M]\right \}= D_M >0,\\
&r:=1+\max\{{\rm diam}(\Omega_T),{\rm diam}(\Omega_R)\},\quad r\leq \min\{D_M/2,D_{TR}/2\}.\nonumber
\end{align}
Remember that we are interested in the study of the operator $(H-z)^{-1}$ for
$z=k_0^2+j0_+$. 

\noindent Denote by
$\chi_{T}$, $\chi_R$ and $\chi_M$ the
characteristic functions of $\Omega_{T}$, $\Omega_R$ and $\Omega_M$ respectively.
Define a smooth function $0\leq J_M\leq 1$ which enters in $\Omega_T$ and $\Omega_R$ and "touches" neither the 
transmitter(s) nor the receiver(s):
\begin{align}\label{apatra10}
 & J_M(\x)=1, \quad \x\not\in \Omega_T\cup\Omega_R,\\
\label{acincea10}&J_M\; \chi_{{\rm supp}(W_T)}=J_M\;\chi_{{\rm supp}(W_R)}=0.
\end{align}
The second condition is possible because the $\Omega$'s completely
contain the two antenna systems. Let us also note the identity:
\begin{align}\label{anoua10}
 J_M(\x)\{1-\chi_T(\x)-\chi_R(\x)\}=1-\chi_T(\x)-\chi_R(\x),\quad \x\in\R^3.
\end{align}

\noindent Since we assumed $r\geq 1$, we can choose a smooth function 
$0\leq J\leq 1$ such that 
\begin{align}\label{asasea10}
 J(\x)=1 \quad {\rm if}\quad |\x|\leq r-1,\qquad J(\x)=0 \quad {\rm if}\quad |\x|\geq r.
\end{align}
Take a point in $\Omega_T$ as the origin of coordinates and let ${\bf e}\in \Omega_R$. Define 
\begin{align}\label{asaptea10}
 J_{T}(\x):=J(\x),\quad J_{R}(\x):=J(\x-{\bf e}).
\end{align}
Our choice of $r$ in \eqref{atreia10} insures that $r-1$ is larger
than the diameters of both $\Omega_T$ and $\Omega_R$. Thus:
\begin{align}\label{aopta10}
 J_{T}\chi_T=\chi_T,\quad J_{R}\chi_R=\chi_R.
\end{align}
The support of any derivative of $J_{T}$ is contained in the annulus 
$r-1\leq |\x|\leq r$, while the support of any derivative of $J_{R}$ is contained in the annulus $r-1\leq |\x-{\bf e}|\leq r$. 

\noindent Let us denote by $\chi_{T_a}$ the characteristic function of the
spherical annulus $r-1\leq |\x|\leq r$, and by $\chi_{R_a}$ the
characteristic function of the spherical annulus $r-1\leq
|\x-{\bf e}|\leq r$. Then we clearly have the identities:
\begin{align}\label{aunspea100}
 \chi_{T_a}(x)D^\alpha J_{T}(\x)=D^\alpha J_{T}(\x),\quad
 \chi_{R_a}(x)D^\alpha J_{R}(\x)=D^\alpha J_{R}(\x).
\end{align}
The supports of $\chi_{T_a}$ and 
$\chi_{R_a}$ are disjoint from each other, and are situated in the
free space.  

\noindent If $z$ has a sufficiently large imaginary part, we can define the
following bounded operator (see \eqref{adoua10}):
\begin{align}\label{adoispea10}
S(z):=J_{T}(H_T-z)^{-1}\chi_T+J_M(H_M-z)^{-1}(1-\chi_T-\chi_R)+J_{R}(H_R-z)^{-1}\chi_R.
\end{align}
We then have:
\begin{align}\label{atreispea10}
 (H-z)S(z)=1+K_T(z)+K_M(z)+K_R(z),
\end{align}
where 
\begin{align}\label{adoua12}
K_T(z)&=[-\Delta,J_{T}](H_T-z)^{-1}\chi_T,\quad 
K_M(z)=[-\Delta,J_M](H_M-z)^{-1}(1-\chi_T-\chi_R),\quad \nonumber \\
K_R(z)&=[-\Delta,J_{R}](H_R-z)^{-1}\chi_R, 
\end{align}
and this is because $W_{T(RM)}$ commute with $J_{T(RM)}$. For example,
$$[W_M,J_M]{\bf A}=\{\nabla \ln(k_M^2)\} \{\nabla J_M \cdot {\bf A}\}=0$$
due to disjoint support properties of $\nabla k_M^2$ and $\nabla J_M$. 

\noindent If $|\Im(z)|$ is large enough, then one can prove that 
$$\max\left \{||K_T||,||K_M||,||K_R||\right \}\leq 1/10,$$
thus 
\begin{align}\label{apaispea10}
 (H-z)^{-1}&=S(z)[1+K_T(z)+K_M(z)+K_R(z)]^{-1}\nonumber \\
&=S(z)-(H-z)^{-1}[K_T(z)+K_M(z)+K_R(z)].
\end{align}
This equality can be extended to all $z$ where both sides make sense,
that is outside of a discrete set of singularities. 

\noindent Let us observe an important symmetry property coming from time
reversal invariance. Define 
\begin{align}\label{asuta1}
\widetilde{W}{\bf A}:=-\nabla \left (\{\nabla \ln(\overline{k^2})\}\cdot {\bf
    A}\right )-\overline{\delta k^2}{\bf A}.
\end{align}
This operator is depending on frequency through $k$, and due to
Assumption \ref{Asump2} and \eqref{aunspea9} we may write an important
identity for the adjoint:
\begin{align}\label{asuta2}
\widetilde{W}^*_{-\omega}=W_\omega.
\end{align}
If $\widetilde{H}_{-\omega}:=-\Delta +\widetilde{W}_{-\omega}$, then another consequence is:
\begin{align}\label{prima12}
 \left \{(\widetilde{H}_{-\omega}-\overline{z})^{-1}\right \}^*=(H_\omega-z)^{-1},
\end{align}
which is also true for $T,R$ and $M$ alone.

\noindent Write \eqref{apaispea10} with  $\overline{z}$ and $-\omega$, then
take the adjoint and use \eqref{prima12}. We obtain:
\begin{align}\label{atreia12}
 (H_\omega-z)^{-1}=\{S_{-\omega}(\overline{z})\}^*-\{K_T(\overline{z},-\omega)+K_M(\overline{z},-\omega)+
K_R(\overline{z},-\omega)\}^*(H_\omega-z)^{-1}.
\end{align}
We introduce the notations:
\begin{align}\label{apatra12}
\tilde{S}(z):=\chi_T
(H_T-z)^{-1}J_{T}+(1-\chi_T-\chi_R)(H_M-z)^{-1}J_M+
\chi_R(H_R-z)^{-1}J_{R},
\end{align}
and 
\begin{align}\label{acincea12}
\tilde{K}_T(z)&:=\{K_T(\overline{z},-\omega)\}^*= 
\chi_T(H_T-z)^{-1}[\Delta,J_{T}], \nonumber \\
\tilde{K}_M(z)&:=\{K_M(\overline{z},-\omega)\}^*=
(1-\chi_T-\chi_R)(H_M-z)^{-1}[\Delta,J_M],
\nonumber \\
\tilde{K}_R(z)&:=\{K_R(\overline{z},-\omega)\}^*= 
\chi_R(H_R-z)^{-1}[\Delta,J_{R}].
\end{align}
Then \eqref{atreia12} can be written in a more compact way:
\begin{align}\label{asasea12}
 (H-z)^{-1}=\tilde{S}(z)-
[\tilde{K}_T(z)+\tilde{K}_M(z)+
\tilde{K}_R(z)](H-z)^{-1}.
\end{align}

\noindent Here is the first of our technical results:
\begin{lemma}\label{lemma1}The following "almost decoupled" formula holds:
  \begin{align}\label{atreia13}
 \chi_{{\rm supp}(W_R)} (H-z)^{-1}\chi_{T}&=
\chi_{{\rm supp}(W_R)}
\tilde{K}_R(z)\cdot \chi_{R_a}(H-z)^{-1}\chi_{T_a} \cdot 
K_T(z)\chi_{T}.
\end{align}
\end{lemma}
\begin{proof}
There are several important things to note here. First, we have the support condition:  
\begin{align}
{\rm supp}(\J_{in})\subset {\rm supp}(\chi_T).
\end{align}
Second, due to the support properties of our various cut-off
functions we have
$K_T^2=K_M^2=K_R^2=K_TK_R=K_RK_T=0$. 
Now if we introduce \eqref{apaispea10} in the left hand side of \eqref{atreia13}, we see
that the term $ \chi_{{\rm supp}(W_R)} S(z) \chi_{{\rm
    supp}(W_T)}=0$ because of various support properties. We obtain:
 \begin{align}\label{prima13}
 \chi_{{\rm supp}(W_R)} (H-z)^{-1}\chi_T=-
\chi_{{\rm supp}(W_R)} (H-z)^{-1}K_T(z)\chi_T.
\end{align}
Now use \eqref{asasea12} in the right hand side of \eqref{prima13}. We
obtain:
 \begin{align}\label{adoua13}
 \chi_{{\rm supp}(W_R)} (H-z)^{-1}\chi_{T}&=-
\chi_{{\rm supp}(W_R)} \tilde{S}(z)K_T(z)\chi_{T}\nonumber \\
&+\chi_{{\rm supp}(W_R)} \tilde{K}_R(z)(H-z)^{-1}K_T(z)\chi_{T}.
\end{align}
But the first term on the rhs of the above equality is zero, again due
to the supports. Finally, use \eqref{aunspea100} in the above equation, and the proof is over.
\end{proof}

\vspace{0.5cm}

{\it Remark}. Note that $\chi_{{\rm supp}(W_R)} (H-z)^{-1}\chi_{T}$ is the operator entering in
\eqref{februar6} giving the general transfer matrix element, because
$\chi_{T}{\bf J}_{\rm in}={\bf J}_{\rm in}$ and 
$\chi_{{\rm supp}(W_R)}\sigma_R=\sigma_R$. Although \eqref{atreia13} 
is an {\bf exact} formula, it is not very
useful yet because in the middle of the right hand side we still have the full
resolvent and not just the resolvent corresponding to the
environment. But in the next subsection we will show that 
$\chi_{R_a}(H-z)^{-1}\chi_{T_a}$ can be better and better approximated
with $\chi_{R_a}(H_M-z)^{-1}\chi_{T_a}$ if the space occupied by the
antennas become smaller and smaller compared to the distances between
different objects. 

\subsection{The study of $\chi_{R_a}(H-z)^{-1}\chi_{T_a} $}

Remember that $\chi_{T_a}$ and $\chi_{R_a}$ are the characteristic functions
of two spherical annuli which are at a distance $r\geq 1$ from both the
transmitter(s) and the receiver(s). If the physical space occupied by
our antennas becomes very small (mathematically this means that the volume
of the supports of $\chi_T$ and $\chi_R$ i.e. of $\Omega_T$ and $\Omega_R$ is very small), 
then it would be natural to be able to
approximate $\chi_{R_a}(H-z)^{-1}\chi_{T_a} $ with 
$\chi_{R_a}(H_M-z)^{-1}\chi_{T_a} $. Let us show this here. 

\begin{lemma}\label{lemma2} Let $v_r:={\rm Vol}(\Omega_R)$ and $v_t:={\rm Vol}(\Omega_T)$. Then outside a discrete set of frequencies, the vector potential 
can be approximated in any $C^k(\Omega_R)$ norm with $k\geq 0$ in the following way:
\begin{equation}\label{asaptea13}
\chi_R\A -\tilde{K}_R(k_0^2+j0_+)\chi_{R_a}[H_M-k_0^2-j0_+]^{-1}\chi_{T_a} 
K_T(k_0^2+j0_+)\J_{\rm in}=o(\max\{v_T,v_r\}).
\end{equation}
\end{lemma}
\begin{proof}
In the Appendix we have formulated a Lippmann-Schwinger type
representation of the total resolvent in the presence of $N$
perturbations, see \eqref{decembrie10}. We want to particularize that
formula for $N=2$ objects. We put $H_0=H_M=-\Delta+W_M$, we take 
$W_T$ to be $W_1$, and $W_R$ will
be $W_2$. The formula \eqref{decembrie10} reads as:
\begin{align}\label{premier1}
R(z)=(H_M-z)^{-1}-\sum_{m=1}^2\sum_{n=1}^2(H_M-z)^{-1}\chi_n A_{nm}(z)\chi_m
(H_M-z)^{-1}.
\end{align}
This implies:
\begin{align}\label{premier2}
&\chi_{R_a}(H-z)^{-1}\chi_{T_a}-\chi_{R_a}(H_M-z)^{-1}\chi_{T_a}\\
&=-\sum_{m=1}^2\sum_{n=1}^2\chi_{R_a}(H_M-z)^{-1}\chi_n A_{nm}(z)\chi_m
(H_M-z)^{-1}\chi_{T_a}.\nonumber
\end{align}
Now the idea is to show that the right hand side of \eqref{premier2}
is small when the volume of the supports of $\chi_1$ and $\chi_2$ are
smaller and smaller.  We need to estimate the norm of the
operator 
$$\chi_{R_a}(H_M-z)^{-1}\chi_n A_{nm}(z)\chi_m
(H_M-z)^{-1}\chi_{T_a}.$$
The factor $\chi_{R_a}(H_M-z)^{-1}\chi_n$ is a Hilbert-Schmidt
operator and its Hilbert-Schmidt norm tends to zero like
$\sqrt{v_{t,r}}$. A similar result holds true for $\chi_m
(H_M-z)^{-1}\chi_{T_a}$. It means that the right hand side of
\eqref{premier2} is close to zero when the linear dimensions of our
antennas become very small. 

\noindent We can thus write:
\begin{align}\label{january31}
 \chi_{R_a}(H-(k_0^2+j0_+))^{-1}\chi_{T_a}=
\chi_{R_a}(H_M-(k_0^2+j0_+))^{-1}\chi_{T_a}
+\mathcal{O}\left(\max\{\sqrt{v_t},\sqrt{v_r} \}\right ),
\end{align}
outside of a discrete set of frequencies. This proves the lemma for the $L^2$ norm; one can actually show that this 
convergence is also true in any $C^k$ norm; the ingredients are  
the elliptic regularity and the fact that the supports of $\chi_T$, $\chi_R$, $\chi_{R_a}$ and $\chi_{T_a}$ 
are disjoint. The proof is over.
\end{proof}

\vspace{0.5cm}

{\it Remark}. An obvious interpretation of this formula is the following: the input current
is transformed into a signal by $K_T$ 
(only depending on the transmitter(s)) and sent 
into the annulus given by $\chi_{T_a}$. From there, $H_M$ scatters the
signal into the observation region of the receiver(s), or
$\chi_{R_a}$. Finally, $\tilde{K}_R$ takes over the signal and sends it to the
receiver(s).  Note that the diameters of $\chi_{T_a}$ and $\chi_{R_a}$
{\it do not have to be large}, and this is exactly what happens when some 
scatterers are close to our antennas. But {\it the linear dimensions
  of the emitting and receveing antennas have to be small} in order to
be sure that we can approximate
$\chi_{R_a}(H-(k_0^2+j0_+))^{-1}\chi_{T_a}$ with 
$\chi_{R_a}(H_M-(k_0^2+j0_+))^{-1}\chi_{T_a}$. 

\noindent Another important observation: if the linear dimensions of the antennas are
important relatively to the other distances in our decomposition, 
then the expression in \eqref{asaptea13} is not correct. We would need
to take into consideration the complete 
formula for the full resolvent \eqref{premier2}, because we can no
longer ignore the higher order coupling between the emitting and receiving
antennas given by the {\bf exact} formula \eqref{premier2}. 

\subsection{The main theorem}

Remember that the transfer matrix in \eqref{februar6} is completely
characterized by the value of ${\bf A}=(H-k_0^2-j0_+)^{-1}{\bf J}_{\rm
  in}$. In this subsection we assume that the antennas are small enough so that
\eqref{asaptea13} makes a good approximation. 
In order to simplify notation, we write $z_0$ instead of
$k_0^2+j0_+$. All the resolvents we have considered until now are
integral operators in the following sense:
\begin{equation}\label{acinspea1}
\{(H-z_0)^{-1}\Psi\}_s(\x)=
\sum_{t=1}^3\int_{\R^3}G_{st}(\x,\x';z_0)\Psi_t(\x')d\x',\quad s\in\{1,2,3\}.
\end{equation} 
Their integral kernels are smooth functions of $\x$ and $\x'$ outside
the diagonal $\x=\x'$, due to general elliptic regularity results. We
will now express the quantity in \eqref{asaptea13} with the help of
the various integral kernels appearing in that equality. Using
\eqref{acincea12} and \eqref{adoua12} we may write (here $\x\in \Omega_R$) : 
\begin{align}\label{acinspea2}
A_{t_1}(\x)&\approx -\sum_{t_2}\int_{\R^3}d {\bf u}G_{t_1t_2}^{(R)}(\x,{\bf
  u};z_0) \sum_{s_1}
\left \{\frac{\partial}{\partial u_{s_1}}\frac{\partial J_R}{\partial
    u_{s_1}}+\frac{\partial J_R}{\partial
    u_{s_1}}\frac{\partial}{\partial u_{s_1}}\right
\}\sum_{t_3}\int_{\R^3} d{\bf v}G_{t_2t_3}^{(M)}({\bf
  u}, {\bf v};z_0)\nonumber \\
&\cdot \sum_{s_2}
\left \{\frac{\partial}{\partial v_{s_2}}\frac{\partial J_T}{\partial
    v_{s_2}}+\frac{\partial J_T}{\partial
    v_{s_2}}\frac{\partial}{\partial v_{s_2}}\right \}
\sum_{t_4}\int_{\R^3}d\x_T G_{t_3t_4}^{(T)}({\bf
  v},\y;z_0)J_{{\rm in},t_4}(\y).
\end{align}
In order to simplify the notation, we will assume that the
transmitters are located near the origin of coordinates,
 while the receivers are located near ${\bf e}$. In this formula we will choose $J_T$ to be a mollifier of 
the characteristic function of the ball $B_{r}(0)=\{|\x|\leq
r\}$, and   $J_R$ to be a mollifier of 
the characteristic function of $B_{r}({\bf e})=\{|\x-{\bf e}|\leq
r\}$. In $B_{r}(0)$ we may choose to work with local spherical
coordinates $(\rho,\hat{\x}_T)$, with $\rho\geq 0$ and 
$\hat{\x}_T\in S^2$. The same thing
can be done for $B_{r}({\bf e})$, and denote its local spherical
coordinates by $(\rho',\hat{\x}_R)$. We choose $J_T$ and $J_R$ 
to be radial in these coordinates. Then as they converge towards the
characteristic functions, one can prove that the expression in
\eqref{acinspea2} converges to:
\begin{align}\label{acinspea3}
A_{t_1}(\x)&\approx
-r^4\sum_{t_2}\int_{S^2}d\hat{\x}_R\int_{\rho'}d\rho'
G_{t_1t_2}^{(R)}(\x,{\bf e}+\rho'\hat{\x}_R;z_0)
\left \{\frac{\partial}{\partial \rho'
  }\delta(\rho'-r)+\delta(\rho'-r)
\frac{\partial}{\partial\rho'}\right
\}\nonumber \\
&\sum_{t_3}\int_{S^2}d\hat{\x}\int_{\rho}d\rho
G_{t_2t_3}^{(M)}({\bf e}+\rho'\hat{\x}_R,
\rho\hat{\x}_T;z_0)\nonumber \\
& \left \{\frac{\partial}{\partial \rho
  }\delta(\rho-r)+\delta(\rho-r)
\frac{\partial}{\partial \rho}\right
\}
\sum_{t_4}\int_{\R^3}d\y G_{t_3t_4}^{(T)}(\rho\hat{\x},\y;z_0)J_{{\rm in},t_4}(\y).
\end{align}
Perform the radial integrals and write $A_{t_1}(\x)$ as a sum of four
terms:
\begin{align}\label{acinspea4}
A_{t_1}(\x)&\approx
-r^4\int_{\R^3}d\y\sum_{t_2,t_3,t_4}\int_{S^2}\int_{S^2}d\hat{\x}_Td\hat{\x}_R
\nonumber \\
&\left \{\{\partial_{\rho'} G_{t_1t_2}^{(R)}(\x;{\bf
    e}+\rho'\hat{\x}_R;z_0)\}_{\rho'=r}
\{\partial_{\rho}G_{t_2t_3}^{(M)}({\bf e}+r\hat{\x}_R;\rho\hat{\x}_T;z_0)\}_{\rho=r}
 G_{t_3t_4}^{(T)}(r\hat{\x}_T;\y;z_0)\right . \nonumber \\
&-\{\partial_{\rho'} G_{t_1t_2}^{(R)}(\x,{\bf e}+\rho'\hat{\x}_R;z_0)\}_{\rho'=r}
G_{t_2t_3}^{(M)}({\bf e}+r\hat{\x}_R;r\hat{\x}_T;z_0)
 \{\partial_{\rho}
 G_{t_3t_4}^{(T)}(\rho\hat{\x}_T;\y;z_0)\}_{\rho=r}\nonumber \\
&+G_{t_1t_2}^{(R)}(\x,{\bf e}+r\hat{\x}_R;z_0)
\{\partial_{\rho'}G_{t_2t_3}^{(M)}({\bf
  e}+\rho'\hat{\x}_R,r\hat{\x}_T;z_0)\}_
{\rho'=r}
 \{\partial_{\rho}
 G_{t_3t_4}^{(T)}(\rho\hat{\x}_T;\y;z_0)\}_{\rho=r}\nonumber \\
&\left .-G_{t_1t_2}^{(R)}(\x;{\bf e}+r\hat{\x}_R;z_0)
\{\partial_{\rho',\rho }^2G_{t_2t_3}^{(M)}({\bf e}+\rho'\hat{\x}_R,
\rho\hat{\x}_T;z_0)\}_
{\rho=\rho'=r}
 G_{t_3t_4}^{(T)}(r\hat{\x}_T;\y;z_0)\right \}\nonumber \\
&J_{{\rm in},t_4}(\y).
\end{align}
Let us introduce the notation (boldface ${\bf G}$ indicates $3\times
3$ matrices):
\begin{equation}\label{acinspea5}
{\bf c}_T(\hat{\x}_T;\y):=\left (\begin{array}{c} \{\partial_{\rho}
 {\bf G}^{(T)}(\rho\hat{\x}_T,\y;z_0)\}_{\rho=r} \\ {\bf
   G}^{(T)}(r\hat{\x}_T;\y;z_0) \end{array}\right ),\;{\bf
c}_R(\hat{\x}_R;\x):=\left (\begin{array}{l} \{\partial_{\rho'}
 {\bf G}^{(R)}(\x,{\bf e}+\rho'\hat{\x}_R;z_0)\}_{\rho'=r} \\ {\bf
   G}^{(R)}(\x;{\bf e}+r\hat{\x}_R;z_0) \end{array}\right ),
\end{equation}
and 
\begin{equation}\label{acinspea6}
\mathcal{M}(\hat{\x}_R;\hat{\x}_T):=\left (\begin{array}{cc}
{\bf G}^{(M)}({\bf e}+r\hat{\x}_R,r\hat{\x}_T;z_0) & -
\{\partial_{\rho}
{\bf G}^{(M)}({\bf e}+r\hat{\x}_R,\rho\hat{\x}_T;z_0)\}_{\rho=r} \\ 
-\{\partial_{\rho'}{\bf G}^{(M)}({\bf
  e}+\rho'\hat{\x}_R,r\hat{\x}_T;z_0)\}_{\rho'=r}
& \{\partial_{\rho',\rho }^2{\bf G}^{(M)}({\bf e}+
\rho'\hat{\x}_R;\rho\hat{\x}_T;z_0)\}_
{\rho,\rho'=r}
\end{array}\right ).
\end{equation}
Then the key equation \eqref{asaptea13} can be rewritten as:

\begin{align}\label{acinspea7}
{\bf A}(\x)&=\left\{(H-z_0)^{-1}{\bf J}_{\rm in}\right \}(\x)\nonumber
\\ 
&\approx r^4
\int d\y \int_{S^2}\int_{S^2}d\hat{\x}_Td\hat{\x}_R\{{\bf
c}_R(\hat{\x}_R;\x)\}^t \mathcal{M}(\hat{\x}_R;\hat{\x}_T){\bf
c}_T(\hat{\x}_T;\y){\bf J}_{\rm in}(\y),
\end{align}
where the transposition operation is considered with respect to the
$2\times 2$ structure. 
 
Already here we can see the complete separation between the
transmitters, receivers, and the rest of the scatterers in the
environment. Introducing this formula back into
\eqref{februar6}, we prove the following theorem:
\begin{theorem}\label{theorem1}In the case in which the linear dimensions of the emitting and receiving antennas 
are small, we can approximate the transfer matrix elements by the formula:
\begin{align}\label{februar7}
\mathcal{H}_{mn}\approx r^4 
\int_{S^2}\int_{S^2}d\hat{\x}_Td\hat{\x}_R \left \langle {\bf
g}_R^{(m)}(\hat{\x}_R) ,\mathcal{M}(\hat{\x}_R;\hat{\x}_T){\bf
g}_T^{(n)}(\hat{\x}_T)\right \rangle ,
\end{align}
where $\langle\cdot ,\cdot \rangle$ denotes the usual dot-product in $\R^6$, while
\begin{align}\label{februar8}
{\bf
g}_T^{(n)}(\hat{\x}_T)&:=\int {\bf
c}_T(\hat{\x}_T;\y){\bf J}_{\rm in}^{(n)}(\y) d\y, \\
{\bf g}_R^{(m)}(\hat{\x}_R) &:=\int_{S_m}
d{\bf S}_m \cdot \left \{ -j \mu_0\omega \{{\bf
c}_R(\hat{\x}_R;\x)\}^t+\nabla \left
  (\frac{1}{j\epsilon\omega+
\sigma_R^{(m)}}\nabla \cdot \{{\bf
c}_R(\hat{\x}_R;\x)\}^t \right)\right \}\sigma_R^{(m)}(\x) \nonumber 
\end{align}
are two six-dimensional vectors only depending on the transmitters and
receivers respectively. They characterize the radiation pattern of our
transmitting and receiving antennas, and they do not change with the
environment.
\end{theorem}

\subsection{Application: the spread function in the case of distant scatterers}

The key mathematical object which completely characterizes how the various
scatterers affect the signal is the Green function of the environment ${\bf
  G}^{(M)}(\x;\x';k_0^2)$. We will try to obtain
some relatively simple formulas for this Green function. We have
already given in the Appendix a fairly
general expression for the resolvent in the presence of $N$
scatterers, see \eqref{decembrie9} and \eqref{decembrie10}. 

We would like to give an even simpler, yet generic expression of this Green
function when the distance between the scatterers and the two balls
$B_r(0)$ and $B_r({\bf e})$ is large enough.  As before,
let us use $z_0=k_0^2$. Assume that the total number of these
scatterers is $S\geq 1$. The free Green function is given by the well 
known formula 
\begin{equation}\label{osuta1}
{\bf G}_{sp}^{\rm free}(\x;\y;k_0^2)=\delta_{sp}\frac{e^{j k_0 |\x-\y|}}{4\pi |\x-\y|},\quad k_0=\frac{|\omega|}{c}.
\end{equation}
Let us denote by $R_0(z_0)$ its
associated operator. 

Then \eqref{decembrie10} gives:
\begin{align}\label{second1}
(H_M-z_0)^{-1}&=R_0(z_0)-\sum_{u=1}^S\sum_{v=1}^SR_0(z_0)\chi_u
A_{uv}(z_0)\chi_vR_0(z_0)\nonumber \\
&=: R_0(z_0)-R_0(z_0)T_M(z_0)R_0(z_0),
\end{align}
Here the operator $T_M=\sum_{u=1}^S\sum_{v=1}^S \chi_u A_{uv}\chi_v$ contains all possible interactions between
various scatterers present in the environment. 

Now assume that the
distance between scatterers and the observation points
$r\hat{\x}_T$ and ${\bf e}+r\hat{\x}_R$ 
is much larger than the radius $r$ (see figure 2). 

\begin{center}

{\rm Figure 2.}

             \includegraphics[scale=0.5]{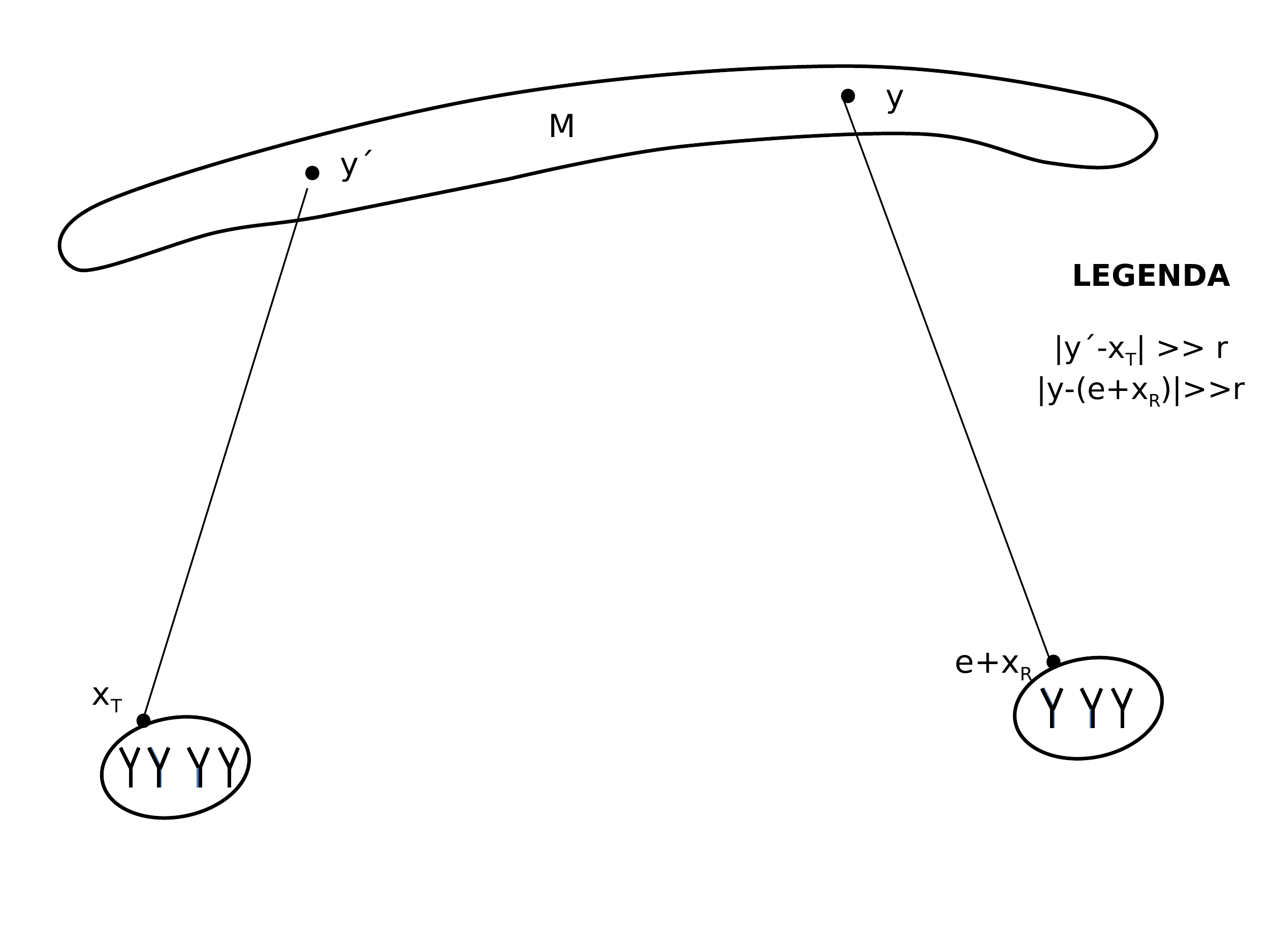} 

\end{center}

Then 
if $\y$ and $\y'$ are in the support of the scatterers  we may write: 
$$|r\hat{\x}_T-\y'|=|\y'|-r\hat{\x}_T\cdot \hat{\y'} +\mathcal{O}(|\y'|^{-1}),\quad 
\hat{\y'}=\frac{\y'}{|\y'|}$$ 
and 
$$|{\bf e}+r\hat{\x}_R-\y|=|\y-{\bf e}|-r\hat{\x}_R\cdot \widehat{(\y-{\bf e})} +
\mathcal{O}(|\y-{\bf e}|^{-1}).$$ 
Therefore
\begin{align}\label{second2}
{\bf G}_{sp}^{\rm free}({\bf e}+r\hat{\x}_R;\y;k_0^2)&=\delta_{sp}
\frac{e^{j k_0 |\y-{\bf e}|}e^{-j k_0 r\hat{\x}_R\cdot
    \widehat{(\y-{\bf e})}}}{4\pi
    |\y-{\bf e}|}+\mathcal{O}(|\y-{\bf e}|^{-2}),\nonumber \\
{\bf G}_{sp}^{\rm free}(\y';r\hat{\x}_T;k_0^2)&=\delta_{sp}
\frac{e^{j k_0 |\y'|}e^{-j k_0 r\hat{\x}_T\cdot \hat{\y'}}}{4\pi
    |\y'|}+\mathcal{O}(|\y'|^{-2}).
\end{align}

Using some elliptic regularity estimates and integration by parts, one can prove the existence of some $3\times 3$ matrices 
${\bf a}_{uv}(\y_u,\y_v;z_0)$ and ${\bf b}_u(\y_u;z_0)$ which are jointly continuous as functions of $\y_u$ and $\y_v$ on the support of $\chi_u$ and 
$\chi_v$ such that:  
\begin{align}\label{second33}
&\{R_0(z_0)\chi_u A_{uv}(z_0)\chi_vR_0(z_0)\}({\bf e}+r\hat{\x}_R,r\hat{\x}_T)\nonumber \\
&=\int_{{\rm supp}(\chi_u)}\int_{{\rm supp}(\chi_v)}
\frac{e^{-j k_0 r\hat{\x}_R\cdot (\widehat {{\bf {y}}_u-{\bf e})}- jk_0r\hat{\x}_T\cdot \hat {{\bf y}}_v}}{16\pi^2
    |{\bf y}_u-{\bf e}|\;|{\bf y}_v|}\left \{{\bf a}_{uv}(\y_u,\y_v;z_0)+{\bf b}(\y_u;z_0)\delta(\y_u-\y_v)\right \}
d\y_ud\y_v \nonumber \\
&+\mathcal{O}(1/D^3_{M}),
\end{align}
where $D_M$ is as in \eqref{atreia10} (i.e. the minimal distance between the emittors/receivers and the scatterers of the environment). The matrices ${\bf a}_{uv}(\cdot,\cdot; z_0)$ and ${\bf b}_{u}(\cdot; z_0)$ contain the full scattering information and depend on all scatterers not just on $u$ and $v$. Note though that the dependence of
$\x$ and $\x'$ is now very explicit. 

Using \eqref{second33} in \eqref{second1} allows us to introduce a 
single $3\times 3$ matrix kernel 
${\bf t}_M(\y,\y';z_0)$ which  contains the full scattering
information 
of the medium:
\begin{align}\label{second4}
&{\bf
  G}^{(M)}({\bf e}+
r\hat{\x}_R,r\hat{\x}_T ;k_0^2)= {\bf G}^{\rm
  free}({\bf e}+
r\hat{\x}_R,r\hat{\x}_T ;k_0^2)\\
&-\int_{{\rm supp}(\chi_M)}\int_{{\rm supp}(\chi_M)}
\frac{e^{-j r k_0
    \hat{\x}_R
\cdot (\widehat {{\bf {y}}-{\bf e})}- 
jr k_0\hat{\x}_T\cdot \hat {{\bf y}'}}}{16\pi^2
    |{\bf y}-{\bf e}|\;|{\bf y}'|}{\bf t}_M(\y,\y';z_0)d{\bf y}d{\bf y}' +\mathcal{O}(1/D^3_{M})\nonumber ,
\end{align}
where $\chi_M$ is the characteristic function of all scatterers, while 
${\bf t}_M(\y,\y';z_0)$ has the structure:
$${\bf t}_M(\y,\y';z_0)=\tau(\y,\y';z_0)+\theta(\y;z_0)\delta(\y-\y')$$
where $\tau(\cdot,\cdot;z_0)$ and  $\theta(\cdot;z_0)$ are continuous on
their domain of definition. These functions play the role of a
scattering kernel, 
and contain the full scattering information of the medium.

When we introduce this formula in \eqref{acinspea6} we obtain two
contributions: one from the "empty" space and another one coming from
the scatterers. The same happens with the transfer matrix.

This scattering generated contribution can be expressed as:
    
\begin{align}\label{februar10}
\mathcal{H}_{mn}^{\rm scatt}\approx r^4 
\int_{S^2}\int_{S^2}d\hat{\x}_Td\hat{\x}_R
\int_{{\rm supp}(\chi_M)}\int_{{\rm supp}(\chi_M)}d{\bf y}d{\bf y}'
\left \langle {\bf
g}_R^{(m)}(\hat{\x}_R) ,\mathcal{M}({\bf y},{\bf y}';\hat{\x}_R;\hat{\x}_T){\bf
g}_T^{(n)}(\hat{\x}_T)\right \rangle
\end{align}
where $\mathcal{M}({\bf y},{\bf y}';\hat{\x}_R;\hat{\x}_T)$ is constructed by
introducing the scattering contribution from \eqref{second4} in
\eqref{acinspea6} which gives:
\begin{align}\label{februar20}
&\mathcal{M}({\bf y},{\bf y}';\hat{\x}_R;\hat{\x}_T):=-\frac{e^{-j r k_0
    \hat{\x}_R
\cdot (\widehat {{\bf y}-{\bf e}})- 
j r k_0\hat{\x}_T\cdot \widehat {{\bf y}'}}}{16\pi^2
    |{\bf y}-{\bf e}|\;|{\bf y}'|}  \nonumber \\
&\times\left (\begin{array}{cc}
{\bf t}_M(\y,\y';z_0) & j  k_0(\hat{\x}_T\cdot \widehat {{\bf y}'}) {\bf t}_M(\y,\y';z_0)\\ 
j k_0
    \hat{\x}_R
\cdot (\widehat {{\bf {y}}-{\bf e}}){\bf t}_M(\y,\y';z_0)
& -k_0^2(\hat{\x}_T\cdot \widehat {{\bf y}'} )(k_0
    \hat{\x}_R
\cdot (\widehat {{\bf {y}}-{\bf e}})){\bf t}_M(\y,\y';z_0)
\end{array}\right ).
\end{align}
We are now ready to formulate the main result of this section:
\begin{corollary} \label{corollary1} The scattering contribution $\mathcal{H}_{mn}^{\rm scatt}$ to the transfer matrix element can be expressed as
\begin{align}\label{februar11}
\mathcal{H}_{mn}^{\rm scatt}= \int\int_{S^2\times
  S^2}d\Omega_R d\Omega_T 
\left \langle {\bf
h}_R^{(m)}(\Omega_R), 
\mathcal{A}(\Omega_R,\Omega_T){\bf
h}_T^{(n)}(\Omega_T)\right \rangle +\mathcal{O}(1/D_M^3),
\end{align}
where ${\bf
h}_T^{(n)}(\Omega_T)$ is a six dimensional vector which can be
interpreted as the signal sent by the transmiter in the direction $\Omega_T$, 
then  $\mathcal{A}(\Omega_R,\Omega_T)$ is a $6\times 6$ spread matrix {\rm only depending on the scatterers}, and finally  
${\bf h}_R^{(m)}(\Omega_R)$ is a six dimensional
vector describing what the receiver got from the direction $\Omega_R$. 
\end{corollary}
\begin{proof}
Going back to \eqref{februar10} we may assume that:
\begin{align}\label{februar21}
{\bf
g}_T^{(n)}(\hat{\x}_T)=\left (\begin{array}{cc}
{\bf F}_1 (\hat{\x}_T)\\ 
{\bf F}_2 (\hat{\x}_T)
\end{array}\right ), \qquad {\bf
g}_R^{(m)}(\hat{\x}_R)=\left (\begin{array}{cc}
{\bf H}_1 (\hat{\x}_R) &
{\bf H}_2 (\hat{\x}_R)
\end{array}\right )
\end{align}
where the ${\bf F}$'s and ${\bf H}$'s are some three dimensional
vectors. Then we have:
\begin{align}\label{februar22}
&\left \langle {\bf
g}_R^{(m)}(\hat{\x}_R) ,\mathcal{M}({\bf y},{\bf y}';\hat{\x}_R;\hat{\x}_T){\bf
g}_T^{(n)}(\hat{\x}_T)\right \rangle=-\frac{e^{-j r k_0
    \hat{\x}_R
\cdot (\widehat {{\bf {y}}-{\bf e}})- 
j r k_0\hat{\x}_T\cdot \widehat {{\bf y}'}}}{16\pi^2
    |{\bf y}-{\bf e}|\;|{\bf y}'|}\nonumber \\
&\times \{{\bf H}_1 (\hat{\x}_R)\cdot {\bf t}_M({\bf y},{\bf y}';z_0){\bf F}_1
(\hat{\x}_T)+{\bf H}_1 (\hat{\x}_R)\cdot {\bf t}_M({\bf y},{\bf y}';z_0){\bf F}_2
(\hat{\x}_T)j  k_0(\hat{\x}_T\cdot \widehat {{\bf y}'})\nonumber \\ 
&+j k_0
    (\hat{\x}_R
\cdot (\widehat {{\bf {y}}-{\bf e}})){\bf H}_2 (\hat{\x}_R)\cdot {\bf t}_M({\bf y},{\bf y}';z_0){\bf F}_1
(\hat{\x}_T)\nonumber \\
&-k_0^2(\hat{\x}_T\cdot \widehat {{\bf y}'})(\hat{\x}_R
\cdot (\widehat {{\bf {y}}-{\bf e}})){\bf H}_2 (\hat{\x}_R)\cdot {\bf t}_M({\bf y},{\bf y}';z_0){\bf F}_2
(\hat{\x}_T)\}.
\end{align}
Now we integrate with respect to the angles $\hat{\x}_R$ and $\hat{\x}_T$ in \eqref{februar10}. It is useful to define local spherical coordinates near both the transmitters and the receivers. Introduce $\Omega_T:=\widehat {{\bf y}'}$ and $s_T:=|{\bf y}'|$, and 
$\Omega_R:=\widehat {{\bf {y}}-{\bf e}}$ and $s_R:=|{\bf {y}}-{\bf e}|$. Now we can define:
\begin{align}\label{februar23}
{\bf
h}_T^{(n)}(\Omega_T):=\left (\begin{array}{cc}
\int_{S^2}e^{-j r k_0\hat{\x}_T\cdot \Omega_T}{\bf F}_1 (\hat{\x}_T)d\hat{\x}_T\\ 
\int_{S^2}e^{-j r k_0\hat{\x}_T\cdot \Omega_T}j  k_0
(\hat{\x}_T\cdot \Omega_T){\bf F}_2 (\hat{\x}_T)d\hat{\x}_T
\end{array}\right ),
\end{align}
\begin{align}\label{februar24}
{\bf
h}_R^{(m)}(\Omega_R):=\left (\begin{array}{cc}
\int_{S^2}e^{-j r k_0\hat{\x}_R
\cdot \Omega_R)}{\bf H}_1 (\hat{\x}_R)d\hat{\x}_R, &
\int_{S^2}e^{-j r k_0\hat{\x}_R
\cdot \Omega_R)}j  k_0(\hat{\x}_R
\cdot \Omega_R)
{\bf H}_2 (\hat{\x}_R)d\hat{\x}_R
\end{array}\right ),
\end{align}
and 
\begin{align}\label{februar25}
\mathcal{A}(\Omega_R,\Omega_T)&:=-\int_0^\infty ds_R\int_0^\infty ds_T\frac{r^4 s_Rs_T}{16\pi^2} \left (\begin{array}{cc}
{\bf t}_M(s_R,\Omega_R; s_T,\Omega_T;z_0) &  {\bf t}_M(s_R,\Omega_R; s_T,\Omega_T;z_0)\\ 
{\bf t}_M(s_R,\Omega_R; s_T,\Omega_T;z_0)
& {\bf t}_M(s_R,\Omega_R; s_T,\Omega_T;z_0)
\end{array}\right ).
\end{align}
Then \eqref{februar11} follows immediately. Note that the integrand in the double integral defining the spread matrix $\mathcal{A}(\Omega_R,\Omega_T)$ is different from zero only on the compact radial supports of the total scatterer defined  by the intersection of the scatterer with the direction $\Omega_T$ seen by an observer placed at the origin and by the intersection of the scatterer with the direction $\Omega_R$ seen by an observer placed at ${\bf e}$. 

\end{proof}

\section{Conclusions and open problems}

\begin{enumerate}
\item Starting from the Maxwell equations and the Ohm's Law we
  give a rigorous analysis 
of the input-output relationship of a MIMO system as a
  direct, well-posed problem. The main observation is that we can  
decouple the group of transmitters from the group of receivers and 
  scatterers if the linear dimensions of all
  our antennas are small enough, see \eqref{asaptea13} and the discussion
  around it. Even if this smallness condition would not be satisfied,
  we could in theory 
  give higher order corrections with respect to the
  decoupled case. 

 We stress that the transmitting (receiving) antennas are allowed to
 interfere among themselves, so in principle the coupling in between the  
transmitting (receiving) antennas is taken into account. 

\item In the decoupled case, we
  can analyze the transfer matrix and identify in it the spread kernel 
  due to the 
  environment alone. The most important result of our paper is
  contained in the equations \eqref{februar7} and \eqref{februar8};
  there we do not need the scatterers to be located in the far-field
  region of the receiving and transmitting antennas. In the particular
  case in
  which the scatterers are far away, then our formulas simplify and we
  can recover previously known empirical results derived in \cite{F1, Po, H,
    Z,S}; in this case, our results are given in 
\eqref{februar23}-\eqref{februar25}. 

 In a future work we will investigate the
  behavior of the spread kernel as a function of angles and frequency for
  scatterers which are not necessarily far away from the emitting and
  receiving antennas. 

\item  Our formalism does not (yet)
  allow ideal metals ($\sigma=\infty$) or discontinuities in
  $\epsilon$'s. Investigating
  how our formalism behaves when one takes the limit of non-smooth
  coefficients is a very
  interesting and difficult problem of operator theory and functional analysis, 
which will be investigated elsewhere.

\item We have not elaborated on the question of the limiting absorption
  principle formulated in \eqref{aunspea14}; complete proofs based on
  the analytic Fredholm alternative will be
  given elsewhere.
  
 \item We think that our new understanding of the construction of the 
 transfer matrix $\mathcal{H}$ paves the way for the study of the behavior of the capacity in the case the number of antennas grows in a definite volume
 (see \cite{Po}  for a discussion on this subject).
  
\end{enumerate}

\section{Appendix: A Lippmann-Schwinger type equation for the resolvent}

Let us consider an operator $H=H_0 +\sum_{k=1}^NW_k$ where $H_0$ is a
"nice" elliptic second order reference differential operator 
(the typical example is $-\Delta$), and the perturbations 
$W_k$ are first order differential operator with smooth
coefficients. The supports of the coefficients of $W_k$ are disjoint
from those of the coefficients of $W_j$ if $k\neq j$. 
Denote by $\chi_m$ the characteristic function of a ball
completely containing the support of the coefficients of $W_m$. 
Denote by $D_{mn}$ the
distance between the supports of $\chi_m$ and $\chi_n$. 

Denote by $R(z):=(H-z)^{-1}$ and by $R_0(z):=(H_0-z)^{-1}$ whenever
the two inverses (resolvents) exist. 

\subsection{The case of just one scatterer}
Assume $N=1$. Choose $z$ in the resolvent set of $H_0$. If the
imaginary part of $z$ is large enough, then $||W_1(H_0-z)^{-1}||<1$
and $(H-z)^{-1}$ exists as a bounded operator in $(L^2(\R^3))^3$. The second resolvent
identity reads as:
\begin{align}\label{decembrie1}
R(z)=R_0(z)-R_0(z)W_1R(z).
\end{align}
Multiply the above equation with $W_1$ at the left, and write
$$W_1R=W_1R_0-W_1R_0W_1R.$$
Using that $\chi_1W_1=W_1$ we have:
 \begin{align}\label{decembrie2}
 W_1R=({\rm Id}_1+W_1R_0\chi_1)^{-1}W_1R_0,
\end{align}
where ${\rm Id}_1$ is the identity operator in $(L^2({\rm
  supp}(\chi_1)))^3$. The inverse $({\rm Id}_1+W_1R_0\chi_1)^{-1}$ 
(if it exists) is
to be taken in $(L^2({\rm supp}(\chi_1)))^3$. 
We will always assume the existence of this inverse;
generically this is true if $W_1$ is relatively compact to $R_0$ and
one can apply the Fredholm alternative. 

Note that if we know $W_1R$, then we know $R$ everywhere in the space
because we can replace \eqref{decembrie2} in \eqref{decembrie1} and
obtain:
\begin{align}\label{decembrie3}
R(z)&=R_0(z)-R_0(z)\chi_1({\rm
  Id}_1+W_1R_0(z)\chi_1)^{-1}W_1R_0(z)\nonumber \\
&=R_0(z)-R_0(z)T_1(z)R_0(z),\nonumber \\
T_1(z)&:=\chi_1({\rm
  Id}_1+W_1R_0(z)\chi_1)^{-1}W_1.
\end{align}

\subsection{The case of several scatterers}
Let $N\geq 2$. The equivalent of \eqref{decembrie1} reads as:
\begin{align}\label{decembrie4}
R(z)=R_0(z)-\sum_{m=1}^NR_0(z)W_mR(z).
\end{align}
We multiply with $W_n$ at the left on both sides of the above equality
and obtain:
\begin{align}\label{decembrie5}
W_nR(z)=W_nR_0(z)-\sum_{m=1}^NW_nR_0(z)W_mR(z)=\sum_{m=1}^N\left \{\delta_{nm}-W_nR_0(z)\chi_m\right\}W_mR(z).
\end{align}
Denote by $\mathcal{H}_N:=\oplus_{k=1}^N(L^2({\rm
  supp}(\chi_k)))^3$. Define the bounded operator
$\mathcal{M}$ given by
\begin{align}\label{decembrie6}
&\mathcal{H}_N\ni \Psi=\oplus_{k=1}^N\psi_k
\to\mathcal{M}(z)\Psi=\oplus_{i=1}^N
\left\{\sum_{k=1}^N\mathcal{M}_{ik}\psi_k\right \},\nonumber \\
&\mathcal{M}_{ik}(z):=W_iR_0(z)\chi_k.
\end{align}
Define the immersion operator
\begin{align}\label{decembrie6'}
J_N: (L^2(\R^3))^3\mapsto \mathcal{H}_N,\quad (L^2(\R^3))^3\ni \phi\to
J\psi=\oplus_{n=1}^N\chi_n\phi,\nonumber \\
J_N^*:\mathcal{H}_N\mapsto (L^2(\R^3))^3,\quad \mathcal{H}_N\ni\Psi=
\oplus_{n=1}^N\psi_n\mapsto J_N^*\Psi=\sum_{n=1}^N\chi_n\psi_n.
\end{align}

Denoting the total perturbation by $W=\sum_{n=1}^NW_n$, the 
equation \eqref{decembrie5} can be seen as:
\begin{align}\label{decembrie7}
J_NWR(z)=J_NWR_0(z)-\mathcal{M}(z)J_NWR(z).
\end{align}
Denote by ${\rm Id}$ the identity operator in $\mathcal{H}_N$. If 
${\rm Id}+\mathcal{M}(z)$ is invertible in $\mathcal{H}_N$, then the
above equation can be rewritten as:
\begin{align}\label{decembrie8}
J_NWR(z)=\left \{{\rm Id}+\mathcal{M}(z)\right \}^{-1}
J_NWR_0(z).
\end{align}
Moreover, equation \eqref{decembrie4} can be written as:
\begin{align}\label{decembrie9}
R(z)&=R_0(z)-R_0(z)J_N^*J_NWR(z)\nonumber \\
&=R_0(z)-R_0(z)J_N^*\left \{{\rm Id}+\mathcal{M}(z)\right \}^{-1}
J_NWR_0(z)
\end{align}
where in the second line we used \eqref{decembrie8}. It follows that
the full resolvent can be always written as:
\begin{align}\label{decembrie10}
R(z)&=R_0(z)-\sum_{n=1}^N\sum_{m=1}^NR_0(z)\chi_n A_{nm}(z)
\chi_mR_0(z),
\end{align}
where the operators $A_{nm}(z)$ act in $(L^2(\R^3))^3$. 

Now let us assume that the distance $D_{mn}$ between any two different
supports is large. We can split the operator $\mathcal{M}$ in a
diagonal and off-diagonal part:
  \begin{align}\label{decembrie11}
\mathcal{M}_d(z):=\oplus_{n=1}^NW_nR_0(z)\chi_n,\quad 
\mathcal{M}_o(z):=\mathcal{M}(z)-\mathcal{M}_d(z).
\end{align}
Clearly, the norm of $\mathcal{M}_o(z)$ becomes smaller and smaller
when $D_{mn}$ becomes larger. For example, if $H_0=-\Delta$ we have
that $||W_mR_0(z)\chi_n||\leq C/D_{mn}$. Thus if the minimal distance
between any two scatterers becomes larger than a critical value, we
can write:
 \begin{align}\label{decembrie12}
\left \{{\rm Id}+\mathcal{M}(z)\right \}^{-1}=\sum_{i=0}^\infty (-1)^i
\left \{{\rm Id}+\mathcal{M}_d(z)\right \}^{-1}
\left \{\mathcal{M}_o(z)[{\rm Id}+\mathcal{M}_d(z]^{-1} )\right\}^{i}.
\end{align}
The operators $[{\rm Id}+\mathcal{M}_d(z)]^{-1} $ are diagonal
and given by:
\begin{align}\label{decembrie13}
[{\rm Id}+\mathcal{M}_d(z)]^{-1}=\oplus_{n=1}^N[{\rm Id}_n+W_nR_0(z)\chi_n]^{-1}.
\end{align}
Making the analogy with \eqref{decembrie3} we introduce the notation
\begin{align}\label{decembrie14}
T_n(z):=[{\rm Id}_n+W_nR_0(z)\chi_n]^{-1}W_n,
\end{align}
where $T_n$ is a bounded operator in $(L^2({\rm
  supp}(\chi_n)))^3$. Then introducing \eqref{decembrie14} and 
\eqref{decembrie12} in \eqref{decembrie9} we obtain (we drop the $z$
dependence for simplicity): 
\begin{align}\label{decembrie15}
R=R_0-\sum_{n=1}^NR_0T_nR_0-\sum_{i=1}^\infty
(-1)^i\tilde{\sum}_{m_0,m_1,...m_{i}=1}^NR_0T_{m_0}R_0T_{m_1}\cdot
\dots\cdot R_0T_{m_i}R_0,
\end{align}
where the symbol $\tilde{\sum}$ means that the sum is performed on
indices which obey $m_0\neq m_1$, $m_1\neq m_2$, ..., $m_{i-1}\neq
m_i$.  

\section{Acknowledgments}

Part of this
work has been done while F.B. was visiting professor at Aalborg
University. H.C. acknowledges support from the Danish FNU grant 
{\it Mathematical Physics} and from the French Embassy in Copenhagen.

\end{document}